\pdfoutput=1

\documentclass[letterpaper, 10 pt, conference]{ieeeconf}  

\IEEEoverridecommandlockouts                              

\usepackage{graphicx,algorithm,algpseudocode,amsthm,amsmath,amsfonts,verbatim,mathtools,enumerate,bm,hyperref,tikz}
\usetikzlibrary{arrows}
\usetikzlibrary{arrows.meta}
\allowdisplaybreaks

\newcommand{\diag}{\mathop{\rm diag}}
\newcommand{\blkdiag}{\mathop{\rm blkdiag}}

\newcommand{\rank}{\mathop{\rm rank}}

\newcommand{\argmin}{\mathop{\rm argmin}}

\newcommand{\norm}[1]{\left\lVert#1\right\rVert}
\newcommand{\mnorm}[1]{{\left\vert\kern-0.25ex\left\vert\kern-0.25ex\left\vert #1 
    \right\vert\kern-0.25ex\right\vert\kern-0.25ex\right\vert}}

\newtheorem{definition}{Definition} 
\newtheorem{theorem}{Theorem}
\newtheorem{lemma}{Lemma}

\newtheorem{assumption}{Assumption}

\newcommand{\eg}{{\it e.g.}}

\definecolor{dgreen}{rgb}{0, 0.42, 0.24}
\definecolor{dblue}{rgb}{0, 0.28, 0.67}

\usepackage{tikz,csvsimple}
\usetikzlibrary {automata,positioning} 
\usepackage{subcaption}
\usepackage{pgfplots,pgfplotstable}
\pgfplotsset{compat=1.9} 

\usepackage{filecontents}

\title{\LARGE \bf Cost Design in Atomic Routing Games
}

\author{Yue~Yu, Shenghui~Chen, David~Fridovich-Keil, and Ufuk~Topcu
\thanks{This research is supported in part by NSF 2211548, NSF 1652113, NASA 80NSSC21M0071, and AFRL FA9550-19-1-0169. Y. Yu, S. Chen, and U. Topcu are with the Oden Institute for Computational Engineering and Sciences, The University of Texas at Austin, TX, 78712, USA (emails:  yueyu@utexas.edu,\, shenghui.chen@utexas.edu,\, utopcu@utexas.edu). D. Fridovich-Keil is with the Department of Aerospace Engineering and Engineering Mechanics, The University of Texas at Austin, TX, 78712, USA (email: dfk@utexas.edu).}%
}

\begin{document}

\maketitle
\thispagestyle{empty}
\pagestyle{empty}

\begin{abstract}
An atomic routing game is a multiplayer game on a directed graph. Each player in the game chooses a path---a sequence of links that connect its origin node to its destination node---with the lowest cost, where the cost of each link is a function of all players' choices. We develop a novel numerical method to design the link cost function in atomic routing games such that the players' choices at the Nash equilibrium minimize a given smooth performance function. This method first approximates the nonsmooth Nash equilibrium conditions with smooth ones, then iteratively improves the link cost function via implicit differentiation. We demonstrate the application of this method to atomic routing games that model noncooperative agents navigating in grid worlds. 
\end{abstract}


\section{Introduction}
\label{sec: introduction}

A routing game is a multiplayer game on a directed graph that contains a collection of nodes and links. Each player in the game chooses a sequence of links, which together form a \emph{path}, that connect its origin node to its destination node---with the lowest cost, where the cost of each link is a function of all players' choices of paths. The Nash equilibrium of this game is a collective path choice where no player can obtain a lower cost by unilaterally switching to an alternative path. Routing games are the fundamental mathematical models for predicting the collective behavior of selfish players in communication and transportation networks \cite{beckmann1956studies,correa2010wardrop,gartner1980optimal1,gartner1980optimal2,roughgarden2007routing,patriksson2015traffic}.     

Cost design---also known as network design---is the problem of designing the link cost function of a routing game so that the Nash equilibrium satisfies certain desired properties, \eg,  matching a desired equilibrium pattern or minimizing the total cost of all players. Cost design is the key to modifying unwanted traffic patterns in congested transportation networks \cite{migdalas1995bilevel,yang1998models,farahani2013review,bertsimas2015data,paccagnan2021optimal,paccagnan2022utility}.

Existing results on cost design are limited to \emph{nonatomic} routing games, where the number of players is assumed to be infinite and each player is a negligible fraction of the entire player population \cite{roughgarden2007routing}. Although reasonable for applications with a large number of players---\eg, predicting the traffic patterns of thousands of vehicles \cite{patriksson2015traffic}---such an assumption is not valid for applications in multi-agent systems with a medium number of agents, where each player is no longer negligible compared with the entire player population.  
 
We develop a numerical method for cost design in \emph{atomic} routing games, where the number of players is finite and each player searches for the shortest path in a directed connected graph. We first show that the Nash equilibrium conditions in atomic routing games are equivalent to a set of nonsmooth piecewise linear equations. Next, we develop an approximate projected gradient method to design the link cost function such that the Nash equilibrium of the game minimizes a given smooth performance function. In each iteration, this method first approximates the nonsmooth Nash equilibrium conditions with a set of smooth nonlinear equations, then updates the link cost function via an approximate projected gradient method based on implicit differentiation. We demonstrate the application of this method to atomic routing games that model noncooperative agents navigating in grid worlds. 

Our results extend a previous cost design method for matrix games \cite{yu2022inverse} to atomic routing games while remaining scalable. In particular, the results in \cite{yu2022inverse} require enumerating all the options of each player. In an atomic routing game, the number of these options equals the number of all possible paths for all players, which scales exponentially with the number of nodes and links of the underlying graph. In contrast, the proposed method does not require such enumeration, and the number of equations needed scales linearly with the number of nodes and links of the underlying graph. 

\paragraph*{Notation} We let \(\mathbb{R}\), \(\mathbb{R}_{\ge 0}\), \(\mathbb{R}_{>0}\), and \(\mathbb{N}\) denote the set of real, nonnegative real, positive real, and natural numbers, respectively. Given \(m, n\in\mathbb{N}\), we let \(\mathbb{R}^n\) and \(\mathbb{R}^{m\times n}\) denote the set of \(n\)-dimensional real vectors and \(m\times n\) real matrices; we let \(\mathbf{1}_n\) and \(I_n\) denote the \(n\)-dimensional vector of all 1's and the \(n\times n\) identity matrix, respectively.  Given positive integer \(n\in\mathbb{N}\), we let \([n]\coloneqq \{1, 2, \ldots, n\}\) denote the set of positive integers less or equal to \(n\). Given \(x\in\mathbb{R}^n\) and \(k\in[n]\), we let \([x]_k\) denote the \(k\)-th element of vector \(x\), and \(\norm{x}\) denote the \(\ell_2\)-norm of \(x\).  Given a square real matrix \(A\in\mathbb{R}^{n\times n}\), we let \(A^\top\), \(A^{-1}\), and \(A^{-\top}\) denote the transpose, the inverse, and the transpose of the inverse of matrix \(A\), respectively; we say \(A\succeq 0\) and \(A\succ 0\) if \(A\) is symmetric positive semidefinite and symmetric positive definite, respectively; we let \(\norm{A}_F\) denote the Frobenius norm of matrix \(A\).  Given continuously differentiable functions \(f:\mathbb{R}^n\to\mathbb{R}\) and \(G:\mathbb{R}^n\to\mathbb{R}^m\), we let \(\nabla_x f(x)\in\mathbb{R}^n\) denote the gradient of \(f\) evaluated at \(x\in\mathbb{R}^n\); the \(k\)-th element of \(\nabla_x f(x)\) is \(\frac{\partial f(x)}{\partial [x]_k}\). Furthermore,  we let \(\partial_x G(x)\in\mathbb{R}^{m\times n}\) denote the Jacobian of function \(G\) evaluated at \(x\in\mathbb{R}^n\); the \(ij\)-th element of matrix \(\partial_x G(x)\) is \(\frac{\partial [G(x)]_i}{\partial [x]_j}\). 
\section{Atomic routing games and its approximation via entropy regularization}
\label{sec: routing}

We first introduce the mathematical model for atomic routing games on a directed graph, followed by a smooth approximation of the Nash equilibria in atomic routing games. These results provide the foundation of the cost design results in the next section.

\subsection{Directed graphs}\label{subsec: graph}
A directed graph \(\mathcal{G}\) contains a set of nodes \( [n]\), a set of directed links \( [m]\). Each link
is an ordered pair of distinct nodes, where the first and second node are the “tail” and “head” of the link, respectively. We characterize the connection among different nodes in graph \(\mathcal{G}\) via the \emph{incidence matrix}, denoted by \(E\in\mathbb{R}^{n\times m}\). The entry \([E]_{ij}\) in matrix \(E\) is associated with node \(i\) and link \(j\) as follows:
\begin{equation}\label{eqn: incidence}
     [E]_{ij}=\begin{cases}
    1, & \text{if node \(i\) is the tail of link \(j\),}\\
    -1, & \text{if node \(i\) is the head of link \(j\),}\\
    0, & \text{otherwise.}
    \end{cases}
\end{equation}

\subsection{Atomic routing games}
\label{subsec: atomic routing}
Given the incidence matrix \(E\in\mathbb{R}^{n\times m}\) of a directed graph \(\mathcal{G}\), we consider a game with \(p\in\mathbb{N}\) players. Player \(i\in[p]\) has an origin node \(o_i\in[n]\) and a destination node \(d_i\in[n]\) in graph \(\mathcal{G}\). We define the key components of this game as follows.
\subsubsection{Players' path choices}
Each player \(i\in[p]\) chooses a path with the lowest cost that connects its origin node \(o_i\) and destination node \(d_i\). We represent the path chosen by player \(i\) via a \emph{flow vector} \(x_i\in\mathbb{R}^m\), where \([x_i]_k\in[0, 1]\) denotes the probability for player \(i\) to use link \(k\) on the path it chooses. 

If vector \(x_i\) denotes a path that connects node \(o_i\) and \(d_i\), then it belongs to certain \emph{feasible flow set}, which we define as follows. Let \(r_i\in\mathbb{R}^n\) denote the \emph{origin-destination vector} of player \(i\) such that \([r_i]_j=1\) if \(j=o_i\); \([r_i]_j=-1\) if \(j=d_i\); and \([r_i]_j=0\) if \(j\neq o_i\) and \(j\neq d_i\). Furthermore, let \(s_i\in\mathbb{R}^{n-1}\) denote the \emph{reduced origin-destination vector} as
\begin{equation}\label{eqn: origin}
    s_i=\Gamma (r_i, d_i),
\end{equation}
where \(\Gamma (r_i, d_i)\in\mathbb{R}^{n-1}\) is vector obtained from vector \(r_i\) by deleting the \(d_i\)-th element in \(r_i\). We also use the notion of \emph{reduced incidence matrix} for player \(i\), defined as follows:
\begin{equation}\label{eqn: reduced incidence}
    E_i\coloneqq \Gamma(E, d_i),
\end{equation}
where \(\Gamma(E, d_i)\in\mathbb{R}^{(n-1)\times m}\) is the matrix obtained from matrix \(E\) by deleting its \(d_i\)-th row. 

Equipped with the above notations, we define the \emph{feasible flow set} for player \(i\), denoted by \(\mathbb{P}_i\subset\mathbb{R}^m\), as 
\begin{equation}\label{eqn: flow polytope}
    \mathbb{P}_i=\{y\in\mathbb{R}^m| E_iy=s_i, y\geq 0\}
\end{equation}
for all \(i\in[p]\). Notice that if \(x_i\in\mathbb{P}_i\) and we let \(e_i^\top\) be the \(d_i\)-th row of matrix \(E\), then \eqref{eqn: incidence} and \eqref{eqn: reduced incidence} together imply that
\begin{equation}\label{eqn: sink flow}
    \begin{aligned}
        e_i^\top x_i = \mathbf{1}_{n}^\top Ex_i-\mathbf{1}_{n-1}^\top E_ix_i=0-\mathbf{1}_{n-1}^\top s_i=-1,
    \end{aligned}
\end{equation}
where the second step is due to the fact that \(\mathbf{1}_{n}^\top E=0_m\). Therefore, \(x_i\in\mathbb{P}_i\) is and only if \(Ex_i=r_i\). In other words, \(x_i\in\mathbb{P}_i\) if and only if \(x_i\) denotes a unit flow that originates at node \(o_i\), disappear at node \(d_i\), and is preserved at any other node.

\subsubsection{The Nash equilibrium conditions}
We assume the cost of each link is a quadratic function of all players' flow vectors, and each player chooses a path that minimizes this quadratic function. In other words, each player \(i\) chooses its flow vector \(x_i\) as follows:
\begin{equation}\label{opt: routing}
       x_i\in\underset{y\in\mathbb{P}_i}{\argmin} \enskip \textstyle \big(b_i+\frac{1}{2}C_{ii} y+\sum_{j\neq i} C_{ij}x_j\big)^\top y,
\end{equation}
where vector \(b_i\in\mathbb{R}^{m}\) defines the \emph{nominal link cost}, which is independent of the players' strategies; matrices \(C_{ij}\in\mathbb{R}^{m\times m}\) for all \(i, j\in[p]\) are matrices such that vector \(C_{ij}x_j\in\mathbb{R}^m\) denotes the link cost for player \(i\) due to its interaction with player \(j\). Here we assume that the link cost is a linear function---rather than general polynomials---of the players' flow vectors to simplify our further computation.

We introduce the notion of Nash equilibrium in an atomic routing game as follows

\begin{definition}\label{def: Nash}
A joint flow \(x\coloneqq \begin{bmatrix}
x_1^\top & \ldots & x_p^\top
\end{bmatrix}^\top\) is a \emph{Nash equilibrium} if \eqref{opt: routing} holds for all \(i\in[p]\).
\end{definition}

\subsection{Computing Nash equilibria via nonlinear programming}

We now discuss how to compute the Nash equilibrium in Definition~\ref{def: Nash}. To this end, we denote the \emph{joint flow} of all players as 
\begin{equation}
    x\coloneqq \begin{bmatrix}
    x_1^\top & x_2^\top & \cdots & x_p^\top
    \end{bmatrix}^\top.
\end{equation}
We also use the following notation:
\begin{equation}\label{eqn: game parameter}
\begin{aligned}
& C \coloneqq  \begin{bsmallmatrix}
    C_{11} & C_{12}  & \hdots & C_{1p}\\
    C_{21} & C_{22}  & \hdots & C_{2p}\\
    \vdots & \vdots & \ddots & \vdots \\
    C_{p1} & C_{p2} & \hdots  & C_{pp}
    \end{bsmallmatrix},\enskip b \coloneqq \begin{bsmallmatrix}
b_1\\
b_2\\
\vdots\\
b_p
\end{bsmallmatrix},\enskip s \coloneqq \begin{bsmallmatrix}
s_1\\ 
s_2\\
\vdots \\
s_p
\end{bsmallmatrix},\\
&E_{[1, p]}\coloneqq \blkdiag(E_1, E_2, \ldots, E_p),
\end{aligned}
\end{equation}
where \(E_{[1, p]}=\blkdiag(E_1, E_2, \ldots, E_n)\) is the block diagonal matrix obtained by aligning matrices \(E_1, E_2, \ldots, E_p\) along the diagonal of matrix \(E_{[1, p]}\).

The following lemma shows how to compute the Nash equilibrium in Definition~\ref{def: Nash} by solving a set of piecewise linear equations.

\begin{lemma}\label{lem: Nash}
Suppose that set \(\mathbb{P}_i\) is nonempty and \(C_{ii}=C_{ii}^\top\succeq 0\) for all \(i\in [p]\). Then \eqref{opt: routing} holds for all \(i\in[p]\) if and only if one of the two following set of conditions holds.
\begin{enumerate}
    \item There exists \(v\in\mathbb{R}^{p(n-1)}\) and \(u\in\mathbb{R}^{pm}\) such that
\begin{subequations}\label{eqn: kkt}
\begin{align}
    s&=E_{[1, p]}x,\\
   u&=b+Cx-E_{[1, p]}^\top v,\enskip u^\top x=0,\enskip x,u\geq 0_{pm}.\label{eqn: complementarity}
\end{align}
\end{subequations}
    \item There exists \(v\in\mathbb{R}^{p(n-1)}\) such that
\begin{subequations}\label{eqn: max eq}
\begin{align}
    s&=E_{[1, p]}x,\\
    0_{pm}&=\min\{x, b+Cx-E_{[1, p]}^\top v\}.\label{eqn: x min}
\end{align}
\end{subequations}
\end{enumerate}
\end{lemma}
\begin{proof}
We start with the conditions in \eqref{eqn: kkt}. Since \(C_{ii}=C_{ii}^\top\succeq 0\), and set \(\mathbb{P}_i\) is nonempty and described by linear constraints only, \eqref{opt: routing} holds if and only if its KKT conditions in \eqref{eqn: kkt} hold \cite[Thm. 27.8]{rockafellar1970convex}.

Next, we prove that the conditions in \eqref{eqn: kkt} are equivalent to those in \eqref{eqn: max eq}. It suffices to show that \eqref{eqn: complementarity} holds for some \(u\in\mathbb{R}^{pm}\) if and only if \eqref{eqn: x min} holds. Notice that \eqref{eqn: complementarity} holds for some \(u\in\mathbb{R}^{pm}\) if and only if
\begin{equation}\label{eqn: k comp}
    [b+Cx-E^\top_{[1, p]}v]_k, [x]_k\geq 0, \, [b+Cx-E^\top_{[1, p]}v]_k[x]_k=0
\end{equation}
for all \(k\in[pm]\). Since the condition in \eqref{eqn: k comp} is equivalent to \(\min\{[x]_k, [b+Cx-E^\top_{[1, p]}v]_k\}=0\), we complete our proof. 
\end{proof}

As a result of Lemma~\ref{lem: Nash}, we can compute the Nash equilibrium in Definition~\ref{def: Nash} by solving either the piecewise linear equations in \eqref{eqn: max eq}, or the following optimization problem with a bilinear objective function
\begin{equation}\label{opt: Nash nlp}
    \begin{array}{ll}
        \underset{x, u, v}{\mbox{minimize}} & u^\top x \\
        \mbox{subject to} & s=E_{[1, p]}x,\\
         & u = b+Cx-E_{[1, p]}^\top v,\enskip  x, u\geq 0_{pm}.
    \end{array}
\end{equation}
In particular, one can verify that the conditions in \eqref{eqn: kkt} hold for some \(x, u\in\mathbb{R}^{pm}\) and \(v\in\mathbb{R}^{p(n-1)}\)---or equivalently, \eqref{eqn: max eq} hold for some \(x\in\mathbb{R}^{pm}\) and \(v\in\mathbb{R}^{p(n-1)}\)---if and only if the optimal value in optimization~\eqref{opt: Nash nlp} equals zero. 

\subsection{Approximate Nash equilibria via entropy-regularization}

We now introduce an approximation of the Nash equilibrium in Definition~\ref{def: Nash}. The idea is to approximate the nonsmooth piecewise linear equations in \eqref{eqn: max eq}--which are difficult to solve in general--with smooth nonlinear ones. To this end, we first introduce the following approximation of optimization~\eqref{opt: routing}:
\begin{equation}\label{opt: entropy}
       x_i\in\underset{y\in\mathbb{P}_i}{\argmin} \enskip \textstyle \big(b_i+\frac{1}{2}C_{ii} y+\sum_{j\neq i} C_{ij}x_j\big)^\top y+\lambda y^\top \ln(y)
\end{equation}
where \(\lambda\in\mathbb{R}_{\ge 0}\) is a nonnegative weight, and \(\ln(y)\) denotes the elementwise natural logarithm of \(y\). 

The optimization in \eqref{opt: entropy} approximates the one in \eqref{opt: routing} by adding an entropy regularization term in the objective function. Similar regularization is common in matrix games \cite{yu2022inverse}. The resulting equilibrium is also known as the \emph{quantal response equilibrium} \cite{mckelvey1998quantal}. 

The following results give a nonlinear-equations-based characterization of the condition in \eqref{opt: entropy}.
\begin{lemma}\label{thm: softmax}
Suppose that set \(\mathbb{P}_i\cap \mathbb{R}_{>0}^m\) is nonempty, \(C_{ii}=C_{ii}^\top\succeq 0\), and \(\lambda>0\). Then \eqref{opt: entropy} holds for all \(i\in[p]\) if and only if there exists \(v\in\mathbb{R}^{p(n-1)}\)  such that
\begin{equation}\label{eqn: softmax}
\begin{aligned}
    s=&E_{[1, p]}x,\\
    x=& \exp(\textstyle \frac{1}{\lambda}(E_{[1, p]}^\top v-b-Cx)-\mathbf{1}_{pm}).
\end{aligned}
\end{equation}
\end{lemma}

\begin{proof}
Since \(C_{ii}\succeq 0\) for all \(i\in[p]\), the quadratic objective function in \eqref{opt: entropy} is a convex function of \(y\). The KKT conditions of the optimization in \eqref{opt: entropy} are given by
\begin{equation}\label{eqn: entropy kkt}
    \begin{aligned}
         &s_i-E_ix_i=0_n,\enskip x_i\geq 0,\\
         & \textstyle b_i+\sum_{j=1}^n C_{ij}x_j+\lambda\ln(x_i)+\lambda\mathbf{1}_m-E_i^\top v_i=0_m,
    \end{aligned}
\end{equation}
for all \(i\in[p]\). Notice that the nonnegativity constraints in set \(\mathbb{P}_i\) are redundant since the logarithm function  implies that \(x_i\in\mathbb{R}_{>0}^m\). Since \(\mathbb{P}_i\cap\mathbb{R}_{>0}^m\) is nonempty, all the linear constraints in optimization~\eqref{opt: entropy} can be satisfied, and we know that \eqref{opt: entropy} holds if and only if \eqref{eqn: entropy kkt} holds for some \(v_i\in\mathbb{R}^n\) \cite[Thm. 27.8]{rockafellar1970convex}. The rest of the proof is based on the relation between logarithm and exponential function and the assumption that \(\lambda>0\).
\end{proof}

In the context of the routing game, Lemma~\ref{thm: softmax} shows the effects of an additional tax in each player's objective function---which corresponds to the entropy term in \eqref{opt: entropy}---on the resulting Nash equilibrium: instead of the nonsmooth equilibrium conditions in \eqref{eqn: max eq}, we obtain the smooth equilibrium conditions in \eqref{eqn: softmax}. 

\subsection{Computing approximate Nash equilibria via nonlinear least-squares}
Lemma~\ref{thm: softmax} shows that solving a smooth approximation of the Nash equilibrium conditions in \eqref{eqn: softmax} is equivalent to solving a set of smooth nonlinear equations, or equivalently, the following nonlinear least-squares problem:
\begin{equation}\label{opt: nonlinear ls}
\begin{array}{ll}
    \underset{x, v}{\mbox{minimize}} & \norm{x-\exp(\frac{1}{\lambda}(E_{[1, p]}^\top v-b-Cx)-\mathbf{1}_{pm})}^2\\
    &+\norm{E_{[1, p]}x-s}^2.
\end{array}
\end{equation}

However, the question remains whether such a solution exists, and if so, whether it is unique or not. To answer these questions, we first make the following assumption (recall the definition of the incidence matrix \(E\) from Section~\ref{subsec: graph}). 

\begin{assumption}\label{asp: main} Set \(\mathbb{P}_i\cap \mathbb{R}_{>0}^m\) is nonempty,
\(\lambda>0\), \(C+C\succeq 0\), and \(C_{ii}=C_{ii}^\top\) for all \(i\in[n]\), and \(\rank E=n-1\).
\end{assumption}

Notice that \(C+C^\top\succeq 0\) and \(C_{ii}=C_{ii}^\top\) together imply that \(C_{ii}\succeq 0\) for all \(i\in[p]\). We note that \(\rank E=n-1\) if \(E\) corresponds to a directed graph obtained by assigning arbitrary directions to the links of a connected undirected graph \cite[Thm. 8.3.1]{godsil2001algebraic}.  

The following theorem provides sufficient conditions on matrix \(C\) that ensure that the solution for optimization~\eqref{opt: nonlinear ls} exists and is unique. The proof is based on the notion of \emph{diagonally strictly concavity} in games \cite{rosen1965existence}.

\begin{theorem}\label{lem: unique}
Suppose that Assumption~\ref{asp: main} holds. There exists unique \(x\in\mathbb{R}^{pm}\) and \(v\in\mathbb{R}^{p(n-1)}\) such that \eqref{eqn: softmax} holds.
\end{theorem}
\begin{proof}
We first prove that there exists a unique \(x\) that satisfies \eqref{eqn: softmax} for some \(v\in\mathbb{R}^{p(n-1)}\). Since \eqref{opt: entropy} implies that \(x_i\) is elementwise strictly positive (due to the logarithm function), Assumption~\ref{asp: main} implies that matrix \(C+C^\top+\diag(x)^{-1}\) is positive definite. Hence one can show that any \(x\) that satisfies \eqref{opt: entropy} for all \(i\in[p]\) is the unique Nash equilibrium of a \(n\)-player diagonally strict concave game, whose existence and uniqueness follows from \cite[Thm. 1]{rosen1965existence} and \cite[Thm. 6]{rosen1965existence}, respectively. Finally, under Assumption~\ref{asp: main}, Lemma~\ref{thm: softmax} states that \(x\) satisfies \eqref{opt: entropy} for all \(i\in[p]\) if and only if there exists \(v\in\mathbb{R}^{p(n-1)}\) such that \eqref{eqn: softmax} holds.

Next, we prove the uniqueness of \(v\) by contradiction. Let \(x\) is the unique vector that that satisfies \eqref{eqn: softmax} for some \(v\in\mathbb{R}^{p(n-1)}\). Suppose that there exists \(v_a\neq  v_b\) such that
\begin{equation*}
\begin{aligned}
     x&=\textstyle \exp(\frac{1}{\lambda}(E_{[1, p]}^\top v_a-b-Cx)-\mathbf{1}_{pm})\\
    &=\textstyle\exp(\frac{1}{\lambda}(E_{[1, p]}^\top v_b-b-Cx)-\mathbf{1}_{pm}),
\end{aligned}
\end{equation*}
Then one can verify that
\begin{equation}\label{eqn: E nullspace}
    E_{[1, p]}^\top (v_a-v_b)=0_{pm}.
\end{equation}
By using the definition of matrix \(E_{[1, p]}\) in \eqref{eqn: game parameter}, we can show that if \eqref{eqn: E nullspace} holds for some \(v_a\neq v_b\), then  there exists \(i\in[n]\) such that  \(E_i^\top z_i=0\) for some \(z_i\in\mathbb{R}^{n-1}, z_i\neq 0_{n-1}\). Hence
\begin{equation}\label{eqn: contradiction}
    \begin{bmatrix}
    z_i^\top & 0
    \end{bmatrix} \begin{bmatrix}
    E_i\\
    e_i
    \end{bmatrix}=0_{m}^\top=\mathbf{1}_n^\top E,
\end{equation}
where \(e_i\) is the \(d_i\)-th row of matrix \(E\), and the second step is due to the definition of matrix \(E\) in \eqref{eqn: incidence}. Since \(z_i\neq 0_{n-1}\), vector \(\mathbf{1}_n^\top\) and \(\begin{bmatrix}
    z_i^\top & 0
    \end{bmatrix}\) are linearly independent.
Based on the definition of \(E_i\) in \eqref{eqn: reduced incidence}, we conclude that there exist two linearly independent vectors in the kernel of matrix \(E^\top\). Therefore, we must have \(\rank E=\rank E^\top \leq n-2\), which contradicts the assumption that \(\rank E=n-1\).
\end{proof}

\section{Cost design via implicit differentiation}
\label{sec: inverse}
We now introduce the cost design problem in atomic routing games. Our task is to design the value of vector \(b\) and matrix \(C\) such that the Nash equilibrium in Definition~\ref{def: Nash} optimizes certain performance. In particular, we consider the following optimization problem for cost design:
\begin{equation}\label{opt: bilevel Nash}
    \begin{array}{ll}
        \underset{x, v, b, C}{\mbox{minimize}} & \psi(x) \\
        \mbox{subject to} & \text{The conditions in \eqref{eqn: max eq} hold, }\\
        &b\in\mathbb{B}, \enskip C\in\mathbb{D}.
    \end{array}
\end{equation}
where \(\psi:\mathbb{R}^{pm}\to\mathbb{R}\) is a continuously differentiable function that evaluates the performance of the Nash equilibrium \(x\), set \(\mathbb{B}\subseteq\mathbb{R}^{pm}\) and \(\mathbb{D}\subseteq \mathbb{R}^{pm\times pm}\) are closed and convex, representing the candidate values for vector \(b\) and matrix \(C\), respectively.

Due to Lemma~\ref{thm: softmax}, one can replace the conditions in \eqref{eqn: max eq} with those in \eqref{eqn: kkt}, and show that optimization~\eqref{opt: bilevel Nash} is a \emph{mathematical program with equilibrium constraints}, a nonconvex nonsmooth optimization problem notoriously difficult to solve \cite{bard2013practical}. To overcome this difficulty, we consider the following approximation to \eqref{opt: bilevel Nash}:
\begin{equation}\label{opt: bilevel entropy}
    \begin{array}{ll}
        \underset{x, v, b, C}{\mbox{minimize}} & \psi(x) \\
        \mbox{subject to} & \text{The conditions in \eqref{eqn: softmax} hold, }\\
        &b\in\mathbb{B}, \enskip C\in\mathbb{D}.
    \end{array}
\end{equation}
The above approximation replaces the nonsmooth conditions in \eqref{eqn: max eq} with the smooth ones in \eqref{eqn: softmax}. As the value of \(\lambda\) decreases, such an approximation becomes more accurate, and a solution for optimization~\eqref{opt: bilevel entropy} becomes a good approximation of a solution for optimization~\eqref{opt: bilevel Nash}.

Next, we discuss how to solve optimization~\eqref{opt: bilevel entropy} using an approximate projected gradient method based on implicitly differentiating the conditions in \eqref{eqn: softmax}.

\subsection{Differentiation through the approximate Nash equilibrium conditions}

We let \(\nabla_b\psi(x)\in\mathbb{R}^{pm}\) and \(\nabla_C\psi(x)\in\mathbb{R}^{pm\times pm}\) be such that
\begin{equation}\label{eqn: b & C grad def}
   \textstyle [\nabla_b\psi(x)]_{i}\coloneqq \frac{\partial \psi(x)}{\partial [b]_i}, \enskip [\nabla_C\psi(x)]_{ij}\coloneqq \frac{\partial \psi(x)}{\partial [C]_{ij}}
\end{equation}
for all \(i,j\in[pm]\). The following theorem provides the numerical formulas to compute the gradient \(\nabla_b\psi(x)\in\mathbb{R}^{pm}\) and \(\nabla_C\psi(x)\in\mathbb{R}^{pm\times pm}\).

\begin{theorem}\label{thm: IFT gradient}
Suppose that Assumption~\ref{asp: main} holds. Let \(x\in\mathbb{R}^{pm}\) and \(v\in\mathbb{R}^{p(n-1)}\) be such that \eqref{eqn: softmax} holds. Let
\begin{subequations}\label{eqn: Jacobian}
\begin{align}
D\coloneqq & \diag(\exp(\textstyle \frac{1}{\lambda}(E_{[1, p]}^\top v-b-Cx)-\mathbf{1}_{pm})),\\
J \coloneqq & \begin{bmatrix}
I_{pm}+\frac{1}{\lambda}DC & \frac{1}{\lambda}DE_{[1, p]}^\top\\
-E_{[1, p]} & 0_{p(n-1)\times p(n-1)}
\end{bmatrix}.
\end{align}
\end{subequations}
If matrix \(J\) is nonsingular, then
\begin{equation}\label{eqn: b & C grad}
\begin{aligned}
    \nabla_b \psi(x) & = \textstyle -\frac{1}{\lambda}\begin{bmatrix} D^\top & 0_{pm\times p(n-1)}\end{bmatrix} J^{-\top}\begin{bmatrix}\nabla_x\psi(x)\\
    0_{p(n-1)}\end{bmatrix},\\
    \nabla_C \psi(x) & = \textstyle -\frac{1}{\lambda}\begin{bmatrix} D^\top & 0_{pm\times p(n-1)}\end{bmatrix} J^{-\top}\begin{bmatrix}\nabla_x\psi(x)\\
    0_{p(n-1)}\end{bmatrix}x^\top.
\end{aligned}
\end{equation}
\end{theorem}

\begin{proof}
Let \(\xi=\begin{bmatrix} x^\top & v^\top\end{bmatrix}^\top\), \(C_r\) denote the \(r\)-th column of matrix \(C\), and
\begin{equation*}
F(\xi, b, C)\coloneqq\begin{bsmallmatrix}
x- \exp( \frac{1}{\lambda}(E_{[1, p]}^\top v-b-Cx)-\mathbf{1}_{pm})\\
    s-E_{[1, p]}x
\end{bsmallmatrix}.
\end{equation*}
Since \(F(\xi, b, C)\) is a continuously differentiable function, the implicit function theorem \cite[Thm. 1B.1]{dontchev2014implicit} implies the following: if \(\partial_{\xi} F(\xi, b, C)\) is nonsingular, then \(\frac{\partial \xi}{\partial b}=-(\partial_\xi F(\xi, b, C))^{-1}\partial_b F(\xi, b, C)\). Using the chain rule we can show that \(\partial_\xi F(\xi, b, C)=J\) and \(\partial_b F(\xi, b, C)=\frac{1}{\lambda}\begin{bmatrix} D^\top & 0_{pm\times p(n-1)}\end{bmatrix}^\top\). Hence we conclude that \(\frac{\partial \xi}{\partial b} =-\frac{1}{\lambda}J^{-1}\begin{bmatrix} D^\top & 0_{pm\times p(n-1)}\end{bmatrix}^\top\). Similarly, we can show that \(\frac{\partial \xi}{\partial C_r} =-\frac{[x]_r}{\lambda}J^{-1}\begin{bmatrix} D^\top & 0_{pm\times p(n-1)}\end{bmatrix}^\top\). The rest of the proof is due to the chain rule.
\end{proof}

Notice that the matrix inverse \(J^{-1}\) in \eqref{eqn: b & C grad} may not exist, or exists but is numerically challenging to compute. As a remedy, we use the following formulas in practice, where \(J^{-1}\) is approximated by \(J^\dagger\), the Moore-Penrose pseudoinverse of matrix \(J\):
\begin{equation}\label{eqn: b & C approx grad}
\begin{aligned}
    \hat{\nabla}_b \psi(x) & = \textstyle -\frac{1}{\lambda}\begin{bmatrix} D^\top & 0_{pm\times p(n-1)}\end{bmatrix} (J^\dagger)^\top\begin{bmatrix}\nabla_x\psi(x)\\
    0_{p(n-1)}\end{bmatrix},\\
    \hat{\nabla}_C \psi(x) & = \textstyle -\frac{1}{\lambda}\begin{bmatrix} D^\top & 0_{pm\times p(n-1)}\end{bmatrix} (J^\dagger)^\top\begin{bmatrix}\nabla_x\psi(x)\\
    0_{p(n-1)}\end{bmatrix}x^\top.
\end{aligned}
\end{equation}

The approximation formulas in \eqref{eqn: b & C approx grad} are well-defined, regardless of the singularity of matrix \(J\). However, it only works as a heuristics method: we do not have theoretical guarantees on the nonsingularity of matrix \(J\). As a result, the gradient may not exist and the formulas in \eqref{eqn: b & C approx grad} only
provides an empirical proxy for the gradients. In practice, however, we find that such an proxy work well. See Section~\ref{sec: numerical} for  further details.

\subsection{Approximate projected gradient method}

We present the approximate projected gradient method for optimization~\eqref{opt: bilevel entropy} in Algorithm~\ref{alg: proj grad}, where the projection map \(\Pi_{\mathbb{B}}:\mathbb{R}^m\to\mathbb{R}^m\) and \(\Pi_{\mathbb{D}}:\mathbb{R}^{m\times m}\to\mathbb{R}^{m\times m}\) is given by
\begin{equation}\label{eqn: proj D}
    \Pi_{\mathbb{B}}(b)  = \underset{z\in\mathbb{B}}{\argmin} \norm{z-b},\enskip 
    \Pi_{\mathbb{D}}(C)  = \underset{X\in\mathbb{D}}{\argmin} \norm{X-C}_F,
\end{equation}
for all \(C\in\mathbb{R}^{m\times m}\). 
At each iteration, this method first solve the nonlinear least-squares problem in \eqref{opt: nonlinear ls}---where \(\lambda\) is a tuning parameter---then update matrix \(C\) using the approximate gradient formulas in \eqref{eqn: b & C approx grad} and a positive step size \(\alpha\) until convergence.

\begin{algorithm}[!ht]
\caption{Approximate projected gradient method. }
\begin{algorithmic}[1]
\Require  \(\psi:\mathbb{R}^{pm}\to\mathbb{R}\),  \(\lambda,\alpha, \delta\in\mathbb{R}_{>0}\), \(k_{\max}\in\mathbb{N}\).
\State Initialize \(b=\mathbf{0}_{pm}\) and \(C=0_{pm\times pm}\).
\While{\(k<k_{\max}\)}
\State Solve optimization~\eqref{opt: nonlinear ls} for \(x\).
\State \(b\gets \Pi_{\mathbb{B}}(b-\alpha \hat{\nabla}_b\psi(x))\)
\State \(C\gets \Pi_{\mathbb{D}}(C-\alpha \hat{\nabla}_C\psi(x))\)
\State \(k\gets k+1\)
\EndWhile
\Ensure Vector \(b\) and matrix \(C\).
\end{algorithmic}
\label{alg: proj grad}
\end{algorithm}

\section{Numerical examples}
\label{sec: numerical}

Although initially designed for optimization~\eqref{opt: bilevel entropy}, Algorithm~\ref{alg: proj grad} also provides good solutions for optimization~\eqref{opt: bilevel Nash} in practice. We demonstrate this phenomenon via numerical examples on atomic routing games on grid worlds.
\subsection{Atomic routing games setup}
We consider atomic routing games in Section~\ref{subsec: atomic routing} defined on 2-dimensional grid worlds. In particular, the directed graph of the game is as follows.
Each node in the graph corresponds to a grid. The two nodes are connected by a link if and only if the corresponding two grids are adjacent. In particular, a \(3\times 3\) grid world corresponds to a graph with \(9\) nodes and \(24\) links; a \(5\times 5\) grid world corresponds to a graph with \(25\) nodes and \(80\) links; see Fig.~\ref{fig: grid} for an illustration. 

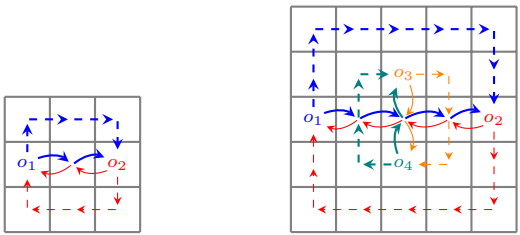
\begin{figure}[!ht]
\centering
\begin{subfigure}[b]{0.49\linewidth}
        \centering
        \begin{tikzpicture}[scale=0.6,shorten >=1pt,node distance=0.5cm,on grid,auto,state/.style={circle,inner sep=0.05pt}]
          \draw[help lines,thick] (0,0) grid (3,3);
          
          \node[state]  (q_1_1) at (0.5,2.5)    {};
          \node[state]  (q_1_2) at (1.5,2.5)    {};
          \node[state]  (q_1_3) at (2.5,2.5)    {};
          \node[state]  (q_2_1) at (0.5,1.5)    {\scriptsize \color{blue}$o_1$};
          \node[state]  (q_2_2) at (1.5,1.5)    {};
          \node[state]  (q_2_3) at (2.5,1.5)    {\scriptsize \color{red}$o_2$};
          \node[state]  (q_3_1) at (0.5,0.5)    {};
          \node[state]  (q_3_2) at (1.5,0.5)    {};
          \node[state]  (q_3_3) at (2.5,0.5)    {};
          
          \path[blue, -stealth, thick] 
                    (q_2_1) edge[bend left]    node    {} (q_2_2)
                    (q_2_2) edge[bend left]    node    {} (q_2_3);
                    
         \path[red, -stealth] 
                    (q_2_3) edge[bend left]    node    {} (q_2_2)
                    (q_2_2) edge[bend left]    node    {} (q_2_1);
                    
         \path[dashed, blue, -stealth, thick] 
                    (q_2_1) edge    node    {} (q_1_1)
                    (q_1_1) edge    node    {} (q_1_2)
                    (q_1_2) edge    node    {} (q_1_3)
                    (q_1_3) edge    node    {} (q_2_3);
                    
         \path[dashed, red, -stealth]        
                    (q_2_3) edge    node    {} (q_3_3)
                    (q_3_3) edge    node    {} (q_3_2)
                    (q_3_2) edge    node    {} (q_3_1)
                    (q_3_1) edge    node    {} (q_2_1);            
        \end{tikzpicture}
        \caption{The two-player game.}
        \label{fig: 3x3 grid}
    \end{subfigure}
    \hfill
    \begin{subfigure}[b]{0.49\linewidth}
        \centering
        \begin{tikzpicture}[scale= 0.6, shorten >=1pt,node distance=0.5cm,on grid,auto,state/.style={circle,inner sep=0.05pt}]
          \draw[help lines,thick] (0,0) grid (5,5);
        
          \node[state]  (q_1_1) at (0.5,4.5)    {};
          \node[state]  (q_1_2) at (1.5,4.5)    {};
          \node[state]  (q_1_3) at (2.5,4.5)    {};
          \node[state]  (q_1_4) at (3.5,4.5)    {};
          \node[state]  (q_1_5) at (4.5,4.5)    {};
          
          \node[state]  (q_2_1) at (0.5,3.5)    {};
          \node[state]  (q_2_2) at (1.5,3.5)    {};
          \node[state]  (q_2_3) at (2.5,3.5)    {\scriptsize\color{orange}$o_3$};
          \node[state]  (q_2_4) at (3.5,3.5)    {};
          \node[state]  (q_2_5) at (4.5,3.5)    {};
          
          \node[state]  (q_3_1) at (0.5,2.5)  {\scriptsize\color{blue}$o_1$};
          \node[state]  (q_3_2) at (1.5,2.5)    {};
          \node[state]  (q_3_3) at (2.5,2.5)    {};
          \node[state]  (q_3_4) at (3.5,2.5)    {};
          \node[state]  (q_3_5) at (4.5,2.5) {\scriptsize\color{red}$o_2$};

          \node[state]  (q_4_1) at (0.5,1.5)    {};
          \node[state]  (q_4_2) at (1.5,1.5)    {};
          \node[state]  (q_4_3) at (2.5,1.5) {\scriptsize\color{teal}$o_4$};
          \node[state]  (q_4_4) at (3.5,1.5)    {};
          \node[state]  (q_4_5) at (4.5,1.5)    {};
          
          \node[state]  (q_5_1) at (0.5,0.5)    {};
          \node[state]  (q_5_2) at (1.5,0.5)    {};
          \node[state]  (q_5_3) at (2.5,0.5)    {};
          \node[state]  (q_5_4) at (3.5,0.5)    {};
          \node[state]  (q_5_5) at (4.5,0.5)    {};
          
          \path[dashed, blue,  -stealth, thick] (q_3_1) edge    node    {} (q_2_1)
                    (q_2_1) edge    node    {} (q_1_1)
                    (q_1_1) edge    node    {} (q_1_2)
                    (q_1_2) edge    node    {} (q_1_3)
                    (q_1_3) edge    node    {} (q_1_4)
                    (q_1_4) edge    node    {} (q_1_5)
                    (q_1_5) edge    node    {} (q_2_5)
                    (q_2_5) edge    node    {} (q_3_5);
                    
            \path[dashed, red,  -stealth] (q_3_5) edge    node    {} (q_4_5)
                    (q_4_5) edge    node    {} (q_5_5)
                    (q_5_5) edge    node    {} (q_5_4)
                    (q_5_4) edge    node    {} (q_5_3)
                    (q_5_3) edge    node    {} (q_5_2)
                    (q_5_2) edge    node    {} (q_5_1)
                    (q_5_1) edge    node    {} (q_4_1)
                    (q_4_1) edge    node    {} (q_3_1);
                    
           \path[dashed, orange,  -stealth] 
                  (q_2_3) edge    node    {} (q_2_4)
                    (q_2_4) edge    node    {} (q_3_4)
                    (q_3_4) edge    node    {} (q_4_4)
                    (q_4_4) edge    node    {} (q_4_3);

           \path[dashed, teal,  -stealth, thick]
                   (q_4_3) edge    node    {} (q_4_2)
                    (q_4_2) edge    node    {} (q_3_2)
                    (q_3_2) edge    node    {} (q_2_2)
                    (q_2_2) edge    node    {} (q_2_3);
                    
          \path[blue,  -stealth, thick] 
                    (q_3_1) edge[bend left]    node    {} (q_3_2)
                    (q_3_2) edge[bend left]    node    {} (q_3_3)
                    (q_3_3) edge[bend left]    node    {} (q_3_4)
                    (q_3_4) edge[bend left]    node    {} (q_3_5);
        
        \path[red,  -stealth]             
                    (q_3_5) edge[bend left]    node    {} (q_3_4)
                    (q_3_4) edge[bend left]    node    {} (q_3_3)
                    (q_3_3) edge[bend left]    node    {} (q_3_2)
                    (q_3_2) edge[bend left]    node    {} (q_3_1);
                    
        \path[orange,  -stealth]             
                    (q_2_3) edge[bend left]    node    {} (q_3_3)
                    (q_3_3) edge[bend left]    node    {} (q_4_3);
                    
        \path[teal,  -stealth, thick]             
                    (q_4_3) edge[bend left]    node    {} (q_3_3)
                    (q_3_3) edge[bend left]    node    {} (q_2_3);
                    
        \end{tikzpicture}
        \caption{The four-player game.}
        \label{fig: 5x5 grid}
    \end{subfigure}
        
    \caption{An illustration of the origin nodes of players and the optimal paths of the players at the Nash equilibrium before (solid) and after cost design (dashed). }
    \label{fig: grid}
\end{figure}

Next, we focus on two cases of atomic routing games on grid worlds: a two-player atomic routing game in a \(3\times 3\) grid world, and a four-player atomic routing game in a \(5\times 5\) grid world. We illustrate the origin grid of each player  in these games along with the paths they choose at the Nash equilibrium in Fig.~\ref{fig: grid}. 

\subsection{Cost design setup}
We make the following choices of parameters in optimization~\eqref{opt: bilevel Nash}. We let 
\begin{equation}\label{eqn: set B & D}
\begin{aligned}
     \mathbb{B}\coloneqq &  \{b\in\mathbb{R}^{pm}| 0_{pm}\leq b\leq \delta\mathbf{1}_{pm}\},\\
     \mathbb{D}\coloneqq & \{C\in\mathbb{R}^{pm\times pm}|  C+C^\top\succeq 0, \enskip \norm{C}_F\leq \rho,\\
     & C_{ii}=C_{ii}^\top,\enskip \forall i\in[p]\},
\end{aligned}
\end{equation}
where matrix \(C\) and its \(i\)-th diagonal block \(C_{ii}\)are associated by \eqref{eqn: game parameter}, \(\delta\in\mathbb{R}_{\ge 0}\) is the initialization parameter in Algorithm~\ref{alg: proj grad}, and \(\rho\in\mathbb{R}_{\ge 0}\) is tuning parameters.
In this case, the projections in \eqref{eqn: proj D} have closed-form formulas; see \cite[Ex. 3.3.15]{bauschke1996projection} and \cite[Lem. 2]{yu2022inverse} for details. For the objective function in optimization~\eqref{opt: bilevel Nash}, we let \(\psi(x)=\frac{1}{2}\norm{x-\hat{x}}^2\),
where \(\hat{x}\in\mathbb{R}^{pm}\) is a desired Nash equilibrium. We illustrate the paths that correspond to \(\hat{x}\) using dashed arrows in Fig.~\ref{fig: grid}.

\subsection{Numerical results}
We demonstrate the application of Algorithm~\ref{alg: proj grad} to the two atomic routing games we constructed. Throughout we let \(\delta=0.1\) and \(\epsilon=0.01\) in Algorithm~\ref{alg: proj grad}; we let \(\overline{x}(b, C)\) denote the solution of optimization~\eqref{opt: Nash nlp} computed by IPOPT \cite{wachter2006implementation} with default accuracy tolerances. Fig.~\ref{fig: lambda} shows the convergence of \(\overline{x}(b, C)\) to the desired equilibrium \(\hat{x}\), where we fix \(\rho= 0.5\) in \eqref{eqn: b & C grad def} and  \(\alpha=0.005\) in Algorithm~\ref{alg: proj grad}. We also illustrate the the cost of the two paths illustrated in Fig.~\ref{fig: grid} before and after the iterations in Algorithm~\ref{alg: proj grad} with \(\lambda=0.005\). These results show that Algorithm~\ref{alg: proj grad} successfully changes the paths of the players at the Nash equilibrium to the desired ones.  Fig.~\ref{fig: bar} shows the path of dashed arrows (gray bars) has a higher cost before design, and a lower cost after design than the paths of solid arrows (white bars), making dashed paths more favorable as desired. 

Although our experiments do not explicitly showcase the differences between the solution of optimization~\eqref{opt: nonlinear ls} and that of optimization~\eqref{opt: Nash nlp}, they confirmed that optimization~\eqref{opt: nonlinear ls} is a valid proxy for optimization~\eqref{opt: Nash nlp} for the purpose of cost design. In particular, Algorithm~\ref{alg: proj grad} designs the value of vector \(b\) and matrix \(C\) by differentiating optimization~\eqref{opt: nonlinear ls}. The resulting vector \(b\) and matrix \(C\), on the other hand, cause the solution of optimization~\eqref{opt: Nash nlp}---denoted by \(\overline{x}(b, C)\), which satisfies the exact Nash equilibrium conditions in \eqref{eqn: kkt}---to match the desired value \(\hat{x}\), as shown by Fig.~\ref{fig: lambda} and Fig.~\ref{fig: rho}.

Optimization~\eqref{opt: entropy} has a numerical limitation. As \(\lambda\) decreases, the numerical values in optimization~\eqref{opt: entropy} increases rapidly. As a result, the value of \(\lambda\) is limited to larger than \(10^{-4}\) in practice to avoid integer overflow. In other words, although optimization~\eqref{opt: entropy} provides an arbitrarily accurate approximation for optimization~\eqref{opt: Nash nlp} as \(\lambda\) decreases, there is a numerical bottleneck of the quality of this approximation.

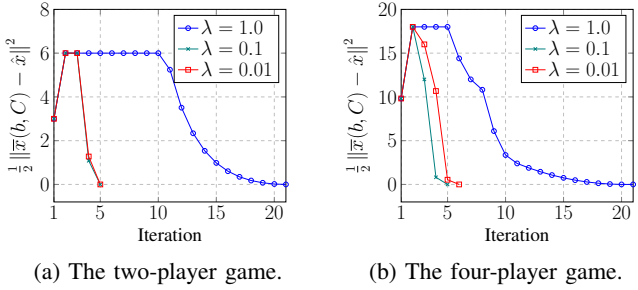
\begin{figure}[!ht]
\centering
    \begin{subfigure}[b]{0.48\linewidth}
    \begin{tikzpicture}[scale=0.45]
    \begin{axis}[
            legend style={font=\LARGE},
            legend cell align={left},
            xlabel={Iteration},
            ylabel={$\frac{1}{2}\norm{\overline{x}(b, C)-\hat{x}}^2$},
            label style={font=\LARGE},
            tick label style={font=\LARGE},
            xmin=1, xmax=21,
            ymax=8,
            xtick={1,5,10,15,20},
            legend pos=north east,
            xmajorgrids=true,
            ymajorgrids=true,
            grid style=dashed,
    ]
    \addplot [blue, mark=o] table [x=iter, y=lambda_1, col sep=comma] {3x3_2p_lambda_plot.csv};
    \addlegendentry{$\lambda=1.0$}
    
    \addplot [teal, mark=x] table [x=iter, y=lambda_0_1, col sep=comma] {3x3_2p_lambda_plot.csv};
    \addlegendentry{$\lambda=0.1$}
    
    \addplot [red, mark=square] table [x=iter, y=lambda_0_0_1, col sep=comma] {3x3_2p_lambda_plot.csv};
    \addlegendentry{$\lambda=0.01$}
    
    
    \end{axis}
    \end{tikzpicture}
    \caption{The two-player game.}
    \end{subfigure}
    \hfill
    \begin{subfigure}[b]{0.48\linewidth}
    \begin{tikzpicture}[scale=0.45]
    \begin{axis}[
            legend style={font=\LARGE},
            legend cell align={left},
            xlabel={Iteration},
            ylabel={$\frac{1}{2}\norm{\overline{x}(b, C)-\hat{x}}^2$},
            label style={font=\LARGE},
            tick label style={font=\LARGE},
            xmin=1, xmax=21,
            ymax=20,
            xtick={1,5,10,15,20},
            legend pos=north east,
            xmajorgrids=true,
            ymajorgrids=true,
            grid style=dashed,
    ]
    \addplot [blue, mark=o] table [x=iter, y=lambda_1, col sep=comma] {5x5_4p_lambda_plot.csv};
    \addlegendentry{$\lambda=1.0$}
    
    \addplot [teal, mark=x] table [x=iter, y=lambda_0_1, col sep=comma] {5x5_4p_lambda_plot.csv};
    \addlegendentry{$\lambda=0.1$}
    
    \addplot [red, mark=square] table [x=iter, y=lambda_0_0_1, col sep=comma] {5x5_4p_lambda_plot.csv};
    \addlegendentry{$\lambda=0.01$}
    
    
    \end{axis}
    \end{tikzpicture}
    \caption{The four-player game.}
    \end{subfigure}
    \caption{The convergence of the iterates in Algorithm~\ref{alg: proj grad}.}\label{fig: lambda}
\end{figure}
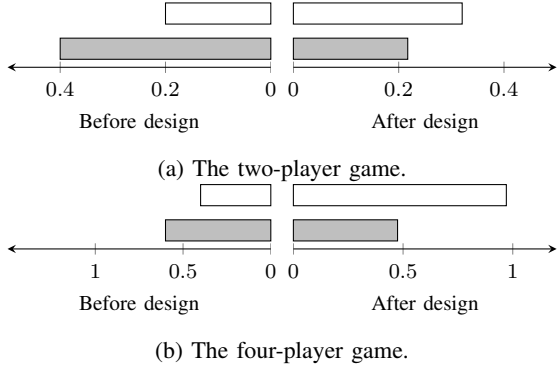
\begin{figure}
    \centering
\begin{subfigure}[b]{\linewidth}
\centering
\begin{tikzpicture}
\begin{axis}[
name=like,
scale only axis,
xbar, xmin=0, xmax=0.5,
xlabel={\footnotesize After design},
width=3.5cm, height= 16pt,
xtick={0,0.2,0.4},
tick label style={font=\footnotesize},
axis x line=left,
axis y line=none,
clip=false,
every axis plot/.append style={
  xbar,
  bar width=0.6,
  bar shift=0pt,
  fill
}
]
\addplot[black,fill=gray!50, yshift=0.2cm] coordinates {
    (0.2172,1)};
\addplot[black,fill=white, yshift=0.2cm] coordinates {
    (0.3206,2)};
\end{axis}
\begin{axis}[
at={(like.north west)},anchor=north east, xshift=-0.3cm,
scale only axis,
xbar, xmin=0, xmax=0.5,
xlabel={\footnotesize Before design},
width=3.5cm, height= 16pt,
x dir=reverse,
xtick={0,0.2,0.4},
tick label style={font=\footnotesize},
axis x line=left,
axis y line=none,
clip=false,
every axis plot/.append style={
  xbar,
  bar width=0.6,
  bar shift=0pt,
  fill
}
]
\addplot[black,fill=gray!50, yshift=0.2cm] coordinates {
    (0.4,1)};
\addplot[black,fill=white, yshift=0.2cm] coordinates {
    (0.2,2)};
\end{axis}
\end{tikzpicture}
\caption{The two-player game.}
\end{subfigure}
\vfill    
\begin{subfigure}[b]{\linewidth}
\centering
\begin{tikzpicture}
\begin{axis}[
name=like,
scale only axis,
xbar, xmin=0, xmax=1.2,
xlabel={\footnotesize After design},
width=3.5cm, height= 16pt,
xtick={0,0.5,1.0},
tick label style={font=\footnotesize},
axis x line=left,
axis y line=none,
clip=false,
every axis plot/.append style={
  xbar,
  bar width=0.6,
  bar shift=0pt,
  fill
}
]
\addplot[black,fill=gray!50, yshift=0.2cm] coordinates {
    (0.4749,1) };
\addplot[black,fill=white, yshift=0.2cm] coordinates {
    (0.9707,2)};
\end{axis}
\begin{axis}[
at={(like.north west)},anchor=north east, xshift=-0.3cm,
scale only axis,
xbar, xmin=0, xmax=1.5,
xlabel={\footnotesize Before design},
width=3.5cm, height= 16pt,
x dir=reverse,
xtick={0,0.5,1.0},
tick label style={font=\footnotesize},
axis x line=left,
axis y line=none,
clip=false,
every axis plot/.append style={
  xbar,
  bar width=0.6,
  bar shift=0pt,
  fill
}
]
\addplot[black,fill=gray!50, yshift=0.2cm] coordinates {
    (0.6,2)};
\addplot[black,fill=white, yshift=0.2cm] coordinates {
    (0.4,3)}; 
\end{axis}
\end{tikzpicture}
\caption{The four-player game.}
   \end{subfigure}
    \caption{The cost of the blue paths in Fig.~\ref{fig: grid} for player 1. Gray and white bars correspond to those marked by dashed and solid arrows, respectively.  }\label{fig: bar}
\end{figure}

We also demonstrate the effects of the parameter \(\rho\) on the Nash equilibrium \(\overline{x}(b, C)\) in Fig.~\ref{fig: rho}, where vector \(b\) and matrix \(C\) are computed by Algorithm~\ref{alg: proj grad} with \(\alpha=\lambda=0.01\). These results confirm an intuition: the larger \(\rho\) is, the more change in matrix \(C\) is allowed in Algorithm~\ref{alg: proj grad}, the closer to the desired value is the Nash equilibrium. 

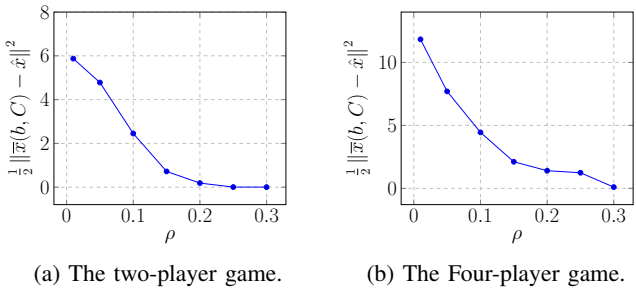
\begin{figure}[!ht]
\centering
    \begin{subfigure}[b]{0.48\linewidth}
     \begin{tikzpicture}[scale=0.45]
        \begin{axis}[
                legend style={font=\LARGE},
                xlabel={$\rho$},
                ylabel={$\frac{1}{2}\norm{\overline{x}(b, C)-\hat{x}}^2$},
                label style={font=\LARGE},
                tick label style={font=\LARGE},
                ymax=8,
                xtick={0,0.1,0.2,0.3},
                legend pos=south east,
                xmajorgrids=true,
                ymajorgrids=true,
                grid style=dashed,
        ]
        \addplot table [x=rho, y=psi_convergence_val, col sep=comma] {3x3_2p_rho_plot.csv};        
        \end{axis}
        \end{tikzpicture}
    \caption{The two-player game.}
    \end{subfigure}
    \hfill
    \begin{subfigure}[b]{0.48\linewidth}
    \begin{tikzpicture}[scale=0.45]
        \begin{axis}[
                legend style={font=\LARGE},
                xlabel={$\rho$},
                ylabel={$\frac{1}{2}\norm{\overline{x}(b, C)-\hat{x}}^2$},
                label style={font=\LARGE},
                tick label style={font=\LARGE},
                ymax=14,
                xtick={0,0.1,0.2,0.3},
                legend pos=south east,
                xmajorgrids=true,
                ymajorgrids=true,
                grid style=dashed,
        ]
        \addplot table [x=rho, y=psi_convergence_val, col sep=comma] {5x5_4p_rho_plot.csv};
        \end{axis}
        \end{tikzpicture}
    \caption{The Four-player game.}
    \end{subfigure}
    \caption{The effects of parameter \(\rho\) on the Nash equilibrium corresponding to the output of Algorithm~\ref{alg: proj grad}. }\label{fig: rho}
\end{figure}

\section{Conclusion}
\label{sec: conclusion}

We developed an approximate projected gradient method to design the link cost function in atomic routing games such that the Nash equilibrium minimizes a given smooth function. However, the current work is limited to linear link cost and routing in determinstic network.
For future work, we consider extensions to routing games with polynomial link cost, stochastic dynamic network routing, and stochastic user equilibrium under noisy cost perception.




\bibliographystyle{IEEEtran}
\bibliography{IEEEabrv,reference}


\end{document}